%% file: root.tex
\documentclass[letterpaper, 10 pt, conference]{ieeeconf}  

\IEEEoverridecommandlockouts                              

\title{\LARGE \bf
Sufficient Conditions for Detectability of Approximately Discretized Nonlinear Systems
}

\author{Seth Siriya, Julian D. Schiller, Victor G. Lopez, and Matthias A. M{\"u}ller
\thanks{The authors are with the Leibniz University Hannover, Institute of Automatic Control, 30167 Hannover (e-mail: {\tt\small \{siriya, schiller, lopez, mueller\}@irt.uni-hannover.de}).}%
\thanks{This work has received funding from the European Research Council (ERC) under the European Union’s Horizon 2020 research and innovation programme (grant agreement No 948679).}
}
\input{extrapackages}

\begin{document}

\maketitle
\thispagestyle{empty}
\pagestyle{empty}

\begin{abstract}

In many sampled-data applications, observers are designed based on approximately discretized models of continuous-time systems, where usually only the discretized system is analyzed in terms of its detectability.
In this paper, we show that if the continuous-time system satisfies certain linear matrix inequality (LMI) conditions, and the sampling period of the discretization scheme is sufficiently small, then the whole family of discretized systems (parameterized by the sampling period) satisfies analogous discrete-time LMI conditions that imply detectability.
Our results are applicable to general discretization schemes, as long as they produce approximate models whose linearizations are in some sense consistent with the linearizations of the continuous-time ones.
We explicitly show that the Euler and second-order Runge-Kutta methods satisfy this condition.
A batch-reactor system example is provided to highlight the usefulness of our results from a practical perspective. 

\end{abstract}

\input{00-intro}
\input{02-problem}
\input{03-results}

\input{04-examples}
\input{05-conclusion}
\appendix
\input{05-proofs}

\addtolength{\textheight}{-11.5cm}   


\bibliography{references.bib}
\bibliographystyle{ieeetr.bst}

\end{document}

%% file: extrapackages.tex
\usepackage{amsmath}
\usepackage{amssymb}
\usepackage{mathtools}
\usepackage{mleftright} 
\usepackage{xcolor}
\usepackage{hyperref}

\usepackage{amsthm}
\newtheorem{theorem}{Theorem}
\newtheorem{assumption}{Assumption}
\newtheorem{definition}{Definition}
\newtheorem{lemma}{Lemma}
\newtheorem{remark}{Remark}
\newtheorem{corollary}{Corollary}
\newtheorem{proposition}{Proposition}

\newcommand{\deriv}[1]{\partial_{#1}}
\newcommand{\Parenth}[1]{\mleft ( #1 \mright )}
\newcommand{\Brace}[1]{\mleft \{ #1 \mright \} }
\newcommand{\Brack}[1]{\mleft [ #1 \mright ] }
\newcommand{\Norm}[1]{\mleft \Vert #1 \mright \Vert}

\DeclareMathOperator{\lambdamin}{\lambda_{\textnormal{min}}}
\DeclareMathOperator{\lambdamax}{\lambda_{\textnormal{max}}}

\mathtoolsset{showonlyrefs}

%% file: 00-intro.tex
\section{Introduction}
Observers are concerned with 
reconstructing the states of a system based on available information from its inputs and outputs.
In many applications, implementing an observer on a digital computer first involves discretizing a continuous-time (CT) system model using a numerical method (e.g., Euler or Runge-Kutta (RK)) with some sampling period, then designing an observer based on the approximate discrete-time (ADT) system.
Stability of the designed observer requires detectability, and so verifying this property is important.

The notion of detectability considered in this work is \textit{incremental input/output-to-state stability} (i-IOSS).
It was first introduced
in \cite{sontag1997output}, where it was shown to be necessary for the existence of a robustly stable observer (compare also \cite{allan2021nonlinear,knufer2020time}), and later also sufficient in \cite{knufer2023nonlinear}.
In recent years, i-IOSS has extensively been employed as a suitable notion for detectability, 
particularly in
the field of optimization-based state estimation (see, e.g., \cite{allan2019lyapunov,allan2021robust,knufer2018robust,knufer2023nonlinear} and \cite{hu2023generic}).
Recently, \cite{schiller2023lyapunov} showed that i-IOSS Lyapunov functions can be constructed for DT nonlinear systems if they satisfy certain LMI conditions, which can be numerically verified in practice (e.g., using gridding, sum-of-squares (SOS) programming \cite{parrilo2003semidefinite,wei2022contraction}, or linear-parameter-varying (LPV) embeddings \cite{koelewijn2021incremental}).
Analogous LMIs can be used to check i-IOSS for CT systems (see \cite{schiller2024robust}), and differential and incremental dissipativity in both the CT and 
discrete-time (DT)
setting (see \cite{verhoek2023convex,koelewijn2021incremental,koelewijn2024convex}).
Note that current DT LMI conditions such as those in \cite{schiller2023lyapunov} (see also \cite{arezki2023lmi}) can already be used to check i-IOSS of an ADT system corresponding to a \textit{specific} sampling period.
However, 
to test i-IOSS for \textit{multiple} sampling periods, 
one
must verify the LMI conditions for each sampling period, which can be time-consuming.
Therefore, it is desirable to obtain conditions under which a whole family (parameterized by sampling period) of ADT systems is i-IOSS.
Since the family of ADT systems is produced from a single CT system, a reasonable condition may involve the CT LMI conditions in \cite{schiller2024robust}. 
This is also desirable since CT dynamics are often simpler than ADT dynamics, resulting in faster verification of the LMIs. However, \cite{schiller2023lyapunov,schiller2024robust,verhoek2023convex,koelewijn2021incremental}, and \cite{koelewijn2024convex}, all focus on either the CT case or DT case separately, and do not investigate what happens during discretization.

Several results 
for
sampled-data systems 
draw connections between CT and DT notions of stability under fast sampling.
In \cite{laila2002open}, it was shown that when a CT controller is designed to satisfy a dissipation inequality for the CT system, the exactly discretized system driven by the corresponding sample-and-hold controller satisfies a similar dissipation inequality in a practical sense.
In \cite{nesic1999sufficient,nesic2004framework}, conditions were provided 
so
a controller that stabilizes an ADT system also practically stabilizes an exactly discretized system.
Analogous results are obtained for sampled-data observers in \cite{arcak2004framework}, which studied the conditions under which stable observers designed for ADT systems estimate the states of the exactly discretized system. However, it focuses on what happens \textit{after} a stable observer for the ADT system has already been successfully designed, and not whether it is even possible in the first place.
On the other hand, \cite{jafarpour2021robust} showed that contractivity of a CT system is preserved in Euler-approximated DT systems.
    
To the best of our knowledge, results 
that investigate verifiable properties of the CT system and approximation scheme that lead to i-IOSS of the associated family of ADT systems are missing from the literature.
Our contributions towards filling this gap are as follows.

Firstly, under a Lipschitz assumption on the CT system, we show that if the LMIs from \cite{schiller2024robust} for i-IOSS can be verified for the CT system, the family of associated ADT systems satisfy the 
DT LMIs for i-IOSS from \cite{schiller2023lyapunov}, as long as the sampling period is sufficiently small. 
Our result is applicable to general approximate discretization schemes, as long as they produce ADT system dynamics whose linearizations are \textit{consistent} with the linearizations of the CT systems. By
``consistent'', 
we mean that the linearizations of the finite-difference approximation of the CT dynamics using the ADT dynamics closely match the linearizations of the actual CT dynamics if the sampling period is sufficiently small.
As a corollary, an i-IOSS Lyapunov function can be constructed for the ADT system.
Secondly, we show that 
Euler and second-order Runge-Kutta 
discretization produces
ADT systems with linearizations that are consistent with those of the CT system, under some regularity conditions on the CT system that are satisfied when the ADT system evolves on a compact set.
Lastly, a batch-reactor system example is provided 
showing
it is 
computationally efficient to verify LMIs for the CT system compared to the ADT system, 
demonstrating
the advantages of our main results. 

\paragraph*{Notation}
Let $\Norm{x}_2$ denote the 2-norm of the vector $x \in \mathbb{R}^n$.
The minimal and maximal eigenvalues of a real symmetric matrix $P \in \mathbb{R}^{n \times n}$ are denoted by $\lambdamin(P)$ and $\lambdamax(P)$ respectively.
The $n \times n$ identity matrix is denoted by $I_n$.
We use $\mathbb{S}_+^{n}$ and $\mathbb{S}_{++}^{n}$ to denote the set of real $n \times n$ symmetric positive semidefinite  and positive definite matrices respectively.
Given a matrix $P \in \mathbb{R}^{n \times m}$, $\Norm{P}_2$ denotes its induced 2-norm.
Given a vector $x \in \mathbb{R}^n$ and a matrix $P \in \mathbb{S}_{++}^n$, $\Norm{x}_P = \sqrt{x^{\top} P x}$ denotes the weighted 2-norm.
A function $\alpha : \mathbb{R}_{\geq 0} \rightarrow \mathbb{R}_{\geq 0}$ is of class $\mathcal{K}$ if it is continuous, strictly increasing, and $\alpha(0) = 0$.

%% file: 02-problem.tex
\section{Problem Setup} \label{sec:problem-setup}
In this section, we first introduce the CT system of interest, followed by the ADT system for which we wish to establish i-IOSS. Then, we introduce the CT and DT LMI conditions that we study in this work.

We consider the CT system
\begin{align}
    \dot{x} = f(x,u,d), \label{eqn:system-dynamics}\\
    y = h(x,u,d), \label{eqn:output-equation}
\end{align}
with state $x \in \mathbb{R}^n$, output $y \in \mathbb{R}^p$, system input $u\in \mathbb{R}^q$, and time-varying parameter $d \in \mathbb{R}^m$.
The mappings $f:\mathbb{R}^n\times\mathbb{R}^q\times\mathbb{R}^m \rightarrow \mathbb{R}^n$ and $h:\mathbb{R}^n\times\mathbb{R}^q\times\mathbb{R}^m\rightarrow \mathbb{R}^p$ represent the system dynamics and output function. 
\begin{remark}
    In the i-IOSS context, it is common to consider $u$ as a ``system input'' and $d$ as ``time-varying parameter''. When i-IOSS is used as a detectability property, $u$ describes the  unknown inputs such as process disturbances and measurement noise, and $d$ the known inputs such as control signals, which is consistent with \cite{schiller2023integral} and \cite{allan2021nonlinear}.
\end{remark}

Next, we assume that a family (parameterized by sampling period $\tau > 0$) of DT models $F_{\tau}: \mathbb{R}^n\times\mathbb{R}^q\times\mathbb{R}^m \rightarrow \mathbb{R}^n$ can be obtained by approximately discretizing the CT system dynamics \eqref{eqn:system-dynamics}.
Examples of discretization methods that can produce such a family of models include the Euler method, the family of Runge-Kutta (RK) methods, and various multistep methods, just to name a few (see \cite{Burden2010} for details). At this stage, we do not specify a particular discretization scheme, given our aim to derive results applicable to a whole class of schemes.
We can then define a family of ADT systems parameterized by sampling period $\tau > 0$ as follows:
\begin{align}
    x^+ &= F_{\tau}(x,u,d), \label{eqn:approx-system-dynamics}\\
    y &= h(x,u,d), \label{eqn:approx-system-output-equation}
\end{align}
where $x^+ \in \mathbb{R}^n$ is the successor state, $x \in \mathbb{R}^n$ is the current state,  $u \in \mathbb{R}^q$ is the system input, $d \in \mathbb{R}^m$ is some time-varying parameter, and $y \in \mathbb{R}^p$ is the system output. Note that the output equation \eqref{eqn:approx-system-output-equation} reuses the output function $h$ from the CT system dynamics \eqref{eqn:output-equation}.

Our aim is to understand whether i-IOSS of the family of ADT systems \eqref{eqn:approx-system-dynamics}-\eqref{eqn:approx-system-output-equation} can be established under checkable conditions on the CT system \eqref{eqn:system-dynamics}-\eqref{eqn:output-equation}, and the scheme used to produce the ADT system. We 
approach this by studying whether 
verifying
the LMI conditions from \cite{schiller2024robust} for the CT system enables verification of 
analogous DT LMI conditions \cite{schiller2023lyapunov} for the ADT systems. 
These conditions are of interest because they enable the construction of i-IOSS Lyapunov functions for CT and DT systems respectively, 
which
are sufficient for i-IOSS (see \cite{schiller2023integral} and \cite{rawlings2017model}).
The remainder of this section is dedicated to introducing the LMI conditions.

Before providing the CT LMI conditions, we first need to assume that the CT system is continuously differentiable.
\begin{assumption} \label{assump:ct-cont-diff}
    The CT system dynamics $f$ \eqref{eqn:system-dynamics} is continuously differentiable.
\end{assumption}
\begin{assumption}
\label{assump:h-cont-diff}
    The output function $h$ (used in \eqref{eqn:output-equation} and \eqref{eqn:approx-system-output-equation}) is continuously differentiable.
\end{assumption}
Under Ass.~\ref{assump:ct-cont-diff}-\ref{assump:h-cont-diff}, the linearizations of $f$ and $h$ with respect to $x$ and $u$ exist, and are denoted by
\begin{align}
    A(x,u,d) := \deriv{x} f(x,u,d), \ B(x,u,d) := \deriv{u} f(x,u,d), \quad \label{eqn:ct-linearized-AB}\\
    C(x,u,d) := \deriv{x} h(x,u,d), \ 
    D(x,u,d) := \deriv{u} h(x,u,d), \quad \label{eqn:ct-linearized-CD}
\end{align}
where $\deriv{x}$ and $\deriv{u}$ denote the Jacobian with respect to $x$ and $u$, respectively.
Whenever it is obvious, we hide the dependency of the linearizations \eqref{eqn:ct-linearized-AB}-\eqref{eqn:ct-linearized-CD} on $(x,u,d)$.
We now provide the CT LMI conditions below.
\begin{definition}[{CT LMI conditions \cite[Ass.~1]{schiller2024robust}}] \label{def:ct-lmi}
    Consider the CT system \eqref{eqn:system-dynamics}-\eqref{eqn:output-equation} satisfying Ass.~\ref{assump:ct-cont-diff}-\ref{assump:h-cont-diff}.
    This system is said to satisfy the \textit{CT LMI conditions} on the set $\mathcal{X}\times\mathcal{U}\times\mathcal{D} \subseteq \mathbb{R}^n\times\mathbb{R}^q\times\mathbb{R}^m$ with Lyapunov matrix $P \in \mathbb{S}_{++}^n$, supply matrices $Q\in\mathbb{S}_{++}^q$ and $R \in \mathbb{S}_{++}^p$, and factor $\kappa > 0$, 
    if for all $(x,u,d) \in {\mathcal{X}}\times{\mathcal{U}}\times{\mathcal{D}}$,
    \begin{align}
        \begin{bmatrix}
            P A + A^{\top} P + \kappa P - C^{\top} R C & P B - C^{\top} R D \\
            B^{\top} P - D^{\top} R C & -D^{\top} R D - Q
        \end{bmatrix} \preceq 0. \quad \label{eqn:thm:ct-lmi-implies-dt-lmi-diff-quot-ct-condition}
    \end{align}
\end{definition}
\begin{remark} \label{rem:finite}
    Note that infinite sets of LMIs such as \eqref{eqn:thm:ct-lmi-implies-dt-lmi-diff-quot-ct-condition} are impossible to verify in general. However, when ${\mathcal{X}}\times{\mathcal{U}}\times{\mathcal{D}}$ is compact,
    gridding can be used to approximately solve it using a finite number of points, transforming the problem into a finite set of LMIs that can be solved via semidefinite programming. Other methods such as sum-of-squares (SOS) optimization and linear-parameter-varying (LPV) embeddings can also be used to transform the problem.
\end{remark}

Next, before providing the LMI conditions for the ADT system, we also assume it is continuously differentiable.
\begin{assumption} \label{assump:dt-diff}
    The ADT system dynamics $F_{\tau}$ \eqref{eqn:approx-system-dynamics} is continuously differentiable 
    for all $\tau > 0$. 
\end{assumption}

Similarly to the CT case, this ensures that the linearizations of the family of ADT system dynamics with respect to $x$ and $u$ exist, and are denoted by
\begin{align}
    \tilde{A}_{\tau}(x,u,d) := \deriv{x} F_{\tau}(x,u,d), \label{eqn:dt-linearized-A}\\
    \tilde{B}_{\tau}(x,u,d) := \deriv{u} F_{\tau}(x,u,d), \label{eqn:dt-linearized-B}
\end{align}
for $\tau > 0$.
Again, whenever it is obvious, we hide the dependency of the linearizations on $(x,u,d)$.

We now define the DT LMI conditions in the context of the ADT systems in \eqref{eqn:approx-system-dynamics}-\eqref{eqn:approx-system-output-equation} studied in this work.
\begin{definition}[{DT LMI conditions \cite[Cor.~3]{schiller2023lyapunov}}]\label{def:dt-lmi}
    Consider the ADT system \eqref{eqn:approx-system-dynamics}-\eqref{eqn:approx-system-output-equation} under Ass.~\ref{assump:h-cont-diff}-\ref{assump:dt-diff} with a specific sampling period $\tau >0$.
    This system is said to satisfy the \textit{DT LMI conditions} on the set $\mathcal{X}\times\mathcal{U}\times\mathcal{D}\subseteq \mathbb{R}^n\times\mathbb{R}^q\times\mathbb{R}^m$ with Lyapunov matrix $P \in \mathbb{S}_{++}^n$, supply matrices $Q \in \mathbb{S}_+^q$ and $R \in \mathbb{S}_+^p$, and factor $\eta \in [0,1)$,
    if for all $(x,u,d) \in {\mathcal{X}}\times{\mathcal{U}}\times{\mathcal{D}}$,
    \begin{align}
        \begin{bmatrix}
            \tilde{A}_{\tau}^{\top} P \tilde{A}_{\tau} - \eta P -  C^{\top} R C & \tilde{A}_{\tau}^{\top} P \tilde{B}_{\tau} - C^{\top}RD \\
            \tilde{B}_{\tau}^{\top} P \tilde{A}_{\tau} - D^{\top} R C & \tilde{B}_{\tau}^{\top} P \tilde{B}_{\tau} - Q - D^{\top} R D
        \end{bmatrix} \preceq 0. \quad \label{eqn:prop:lyapunov-function-lmi-condition}
    \end{align}
\end{definition}

%% file: 03-results.tex
\section{Results} \label{sec:results}

Sec.~\ref{sec:detectability-class} contains the main result of this work. It says that under some regularity conditions,
verifying the CT LMIs in \eqref{eqn:thm:ct-lmi-implies-dt-lmi-diff-quot-ct-condition} for the CT system \eqref{eqn:system-dynamics}-\eqref{eqn:output-equation} allows us to establish that the family of ADT systems from \eqref{eqn:approx-system-dynamics}-\eqref{eqn:approx-system-output-equation} satisfies the DT LMIs in \eqref{eqn:prop:lyapunov-function-lmi-condition} 
with sufficiently small sampling period.
This relies on 
the assumption that 
the linearizations of the parameterized family of ADT systems in \eqref{eqn:dt-linearized-A}-\eqref{eqn:dt-linearized-B} are \textit{consistent} with those of the CT system in \eqref{eqn:ct-linearized-AB}. 
In Sec.~\ref{sec:i-ioss-lyapunov}, we highlight that the LMI result can be used to construct an i-IOSS Lyapunov function for the family of ADT systems. Then, in Sec.~\ref{sec:consistency-euler-rk2}, we show that the Euler and second-order Runge-Kutta (RK2) discretization schemes yield consistent linearizations under 
regularity conditions on the CT system. We also show that these conditions are satisfied when the ADT system evolves on compact set, which is often not a limiting requirement in practice. All proofs are deferred to the appendix.

\subsection{Preservation of LMI conditions after approximate discretization} \label{sec:detectability-class}

We start this section by first defining the consistency of linearizations as follows.
\begin{definition}[Consistency of linearizations] \label{def:consistent-lineariztions}
    The family of linearizations $\{(\tilde{A}_{\tau},\tilde{B}_{\tau})\}_{\tau > 0}$ \eqref{eqn:dt-linearized-A}-\eqref{eqn:dt-linearized-B} of the ADT system \eqref{eqn:approx-system-dynamics}-\eqref{eqn:approx-system-output-equation} under Ass.~\ref{assump:dt-diff} is said to be consistent with the linearizations $(A,B)$ \eqref{eqn:ct-linearized-AB} of the CT system \eqref{eqn:system-dynamics}-\eqref{eqn:output-equation} under Ass.~\ref{assump:ct-cont-diff}
    on the set $\mathcal{X}\times\mathcal{U}\times\mathcal{D}\subseteq\mathbb{R}^n\times\mathbb{R}^q\times\mathbb{R}^m$
    if there exist $\rho \in \mathcal{K}$ and $\tau_0 > 0$ such that 
    \begin{align}
        \max \Brace{\Norm{\frac{\tilde{A}_{\tau}-I_n}{\tau} - A}_2, \Norm{\frac{\tilde{B}_{\tau}}{\tau} - B}_2 } \leq \rho(\tau) \label{eqn:thm:ct-lmi-implies-dt-lmi-diff-quot-diff-condition}
    \end{align}
    for all $(x,u,d) \in {\mathcal{X}}\times{\mathcal{U}}\times{\mathcal{D}}$ and $\tau \in (0,\tau_0]$.
\end{definition}

\begin{remark}
    Consistency in Def.~\ref{def:consistent-lineariztions} can be interpreted as saying that the linearizations of the finite-difference approximation of the CT dynamics $f$ based on the ADT dynamics $F_{\tau}$ (specifically $(F_{\tau}(x,u,d) - x)/\tau$) closely matches the linearizations of the CT dynamics, when $\tau$ is small. 
    This is seen
    by noticing that $\deriv{x}(F_{\tau}(x,u,d) - x)/\tau = (\tilde{A}_{\tau}-I)/\tau$ and $\deriv{u}(F_{\tau}(x,u,d) - x)/\tau = \tilde{B}_{\tau}/\tau$, recalling the definitions of the CT system linearizations in \eqref{eqn:ct-linearized-AB}, and using \eqref{eqn:thm:ct-lmi-implies-dt-lmi-diff-quot-diff-condition}.    
    A similar notion of consistency is used in \cite{nesic1999sufficient}, \cite{nesic2004framework}, and \cite{arcak2004framework} which directly relates the dynamics of the ADT system and exactly discretized system, rather than involving linearizations.
\end{remark}

We now provide our main result in Thm.~\ref{thm:ct-implies-dt}. It says that under a Lipschitz condition on the CT system, and assuming consistency of the linearizations, the family of ADT systems satisfies the DT LMI conditions when the CT LMI conditions are satisfied and the sampling period is small.
\begin{theorem} \label{thm:ct-implies-dt}
    Consider a set $\mathcal{X}\times\mathcal{U}\times\mathcal{D}\subseteq\mathbb{R}^n\times\mathbb{R}^q\times\mathbb{R}^m$, the CT system \eqref{eqn:system-dynamics}-\eqref{eqn:output-equation}, and the family of associated approximate DT systems \eqref{eqn:approx-system-dynamics}-\eqref{eqn:approx-system-output-equation} under Ass.~\ref{assump:ct-cont-diff}-\ref{assump:dt-diff}.
    Suppose the following conditions are satisfied:
    \begin{enumerate}
        \item There exists $L_f > 0$ such that $f$ in \eqref{eqn:system-dynamics} is $L_f$-Lipschitz on ${\mathcal{X}}\times{\mathcal{U}}\times{\mathcal{D}}$, i.e., $\sup_{(x,u,d)\in {\mathcal{X}}\times{\mathcal{U}}\times{\mathcal{D}}} \Norm{\begin{bmatrix}
            A & B
        \end{bmatrix}}_2 \leq L_f $, with $(A,B)$ in \eqref{eqn:ct-linearized-AB};
        \item There exist $\rho\in \mathcal{K}$ and $\tau_0 > 0$ such that $\{(\tilde{A}_{\tau},\tilde{B}_{\tau})\}_{\tau > 0}$ from \eqref{eqn:dt-linearized-A}-\eqref{eqn:dt-linearized-B} is consistent with $(A,B)$ from \eqref{eqn:ct-linearized-AB}
        on ${\mathcal{X}}\times{\mathcal{U}}\times{\mathcal{D}}$ 
        in the sense of Def.~\ref{def:consistent-lineariztions};
        \item There exist $P \in \mathbb{S}_{++}^n$, $Q \in \mathbb{S}_{++}^q$, $R \in \mathbb{S}_{++}^p$, and $\kappa > 0$, such that \eqref{eqn:system-dynamics}-\eqref{eqn:output-equation} satisfies the CT LMI conditions on ${\mathcal{X}}\times{\mathcal{U}}\times{\mathcal{D}}$ 
        in the sense of Def.~\ref{def:ct-lmi}.
    \end{enumerate}
    Then, the family of ADT systems \eqref{eqn:approx-system-dynamics}-\eqref{eqn:approx-system-output-equation} with sampling period $\tau \in \Parenth{0,\tau_1}$ satisfy the DT LMI conditions \eqref{eqn:prop:lyapunov-function-lmi-condition} 
    on ${\mathcal{X}}\times{\mathcal{U}}\times{\mathcal{D}}$, 
    with Lyapunov matrix $P\in \mathbb{S}_{++}^n$, supply matrices $\tilde{Q}(\tau) \in \mathbb{S}_{++}^q$ and $\tilde{R}(\tau) \in \mathbb{S}_{++}^p$, and factor $\eta(\tau) \in (0,1)$, 
    where
    \begin{align}
        \tau_1 &:= \min \Brace{ \frac{1}{\kappa} , \tau_0, \alpha^{-1}(\kappa) }, \label{eqn:theorem:ct-implies-dt-tau1} \\
        \tilde{Q}(\tau) &:= \tau \Parenth{Q + \alpha(\tau) \lambdamin(P) I_q}, \label{eqn:theorem:ct-implies-dt-Qtilde} \\
        \tilde{R}(\tau) &:= \tau R, \label{eqn:theorem:ct-implies-dt-Rtilde} \\
        \eta(\tau) &:= \tau \Parenth{ \alpha(\tau) - \kappa } + 1, \label{eqn:theorem:ct-implies-dt-eta} \\
        \alpha(\tau) &:=  \frac{\lambdamax(P)}{\lambdamin(P)}  \Parenth{4 \rho(\tau) \Parenth{1 + \tau L_f + \tau \rho(\tau)} + \tau L_f^2}. \quad  \label{eqn:theorem:ct-implies-dt-alpha1}
    \end{align}
\end{theorem}
\noeqref{eqn:theorem:ct-implies-dt-Qtilde} 
\noeqref{eqn:theorem:ct-implies-dt-Rtilde}
We again emphasize that an explicit discretization method has not been assumed, 
so
Thm.~\ref{thm:ct-implies-dt} is applicable to a wide range of techniques. 
In order 
to 
apply Thm.~\ref{thm:ct-implies-dt} to a specific discretization method, one 
needs to 
show that the linearizations of the family of ADT system dynamics and the original CT system are consistent.
Later on in Sec.~\ref{sec:consistency-euler-rk2}, we show this explicitly for Euler and RK2 discretization.

\begin{remark} \label{rem:lipschitz}
    Requiring that $f$ is globally Lipschitz continuous in condition~1 of Thm.~\ref{thm:ct-implies-dt} seems restrictive. However, it is satisfied when ${\mathcal{X}}$, ${\mathcal{U}}$, and ${\mathcal{D}}$ are compact due to Ass.~\ref{assump:ct-cont-diff}. This is 
    relevant when verifying LMIs via gridding (see Rem.~\ref{rem:finite}).
\end{remark}

\begin{remark} \label{rem:explicit}
    To construct the parameters/functions in  \eqref{eqn:theorem:ct-implies-dt-tau1}-\eqref{eqn:theorem:ct-implies-dt-alpha1}, aside from requiring $P$, $Q$, $R$, and $\kappa$ (obtained by verifying Def.~\ref{def:ct-lmi}), knowledge of $L_f$, $\tau_0$, and $\rho$, is also required. Later on, in Sec.~\ref{sec:consistency-euler-rk2}, we explicitly provide these quantities when the Euler and RK2 methods are used for discretization.
\end{remark}

\subsection{i-IOSS Lyapunov functions for approximately discretized systems} \label{sec:i-ioss-lyapunov}

Although hinted at previously, in this section, we will formally show that satisfaction of the conditions in Thm.~\ref{thm:ct-implies-dt} allow for the establishment of a i-IOSS Lyapunov function for the parameterized family of ADT systems. Before doing this, we formally recall relevant theory from the literature.

Firstly, we define (exponential) i-IOSS Lyapunov functions for ADT systems with a \textit{specific} sampling period.
\begin{definition}[{i-IOSS Lyapunov function \cite[Ass.~1]{schiller2024robust}}] \label{def:dt-lyapunov}
    Consider an ADT system \eqref{eqn:approx-system-dynamics}-\eqref{eqn:approx-system-output-equation} 
    with specific sampling period $\tau > 0$.
    Given a matrix $P \in \mathbb{S}_{++}^n$, the function $W(x,\tilde{x}):=\Norm{x-\tilde{x}}_P^2$ is an (exponential) \textit{i-IOSS Lyapunov function} on $\mathcal{X}\times\mathcal{U}\times\mathcal{D}\subseteq\mathbb{R}^n\times\mathbb{R}^q\times\mathbb{R}^m$ with supply matrices $Q\in\mathbb{S}_+^q$ and $R \in \mathbb{S}_+^p$, and factor $\eta \in [0,1)$, for this system, if
    \begin{align}
        &W\Parenth{ F_{\tau}(x,u,d), F_{\tau}(\tilde{x},\tilde{u},d) } \\
        &\leq \eta W(x,\tilde{x})+ \Norm{u - \tilde{u}}_Q^2 + \Norm{h(x,u,d) - h(\tilde{x},\tilde{u},d)}_R^2, \label{eqn:lyapunov-condition}
    \end{align}
    for all $d \in {\mathcal{D}}$, $x,\tilde{x} \in {\mathcal{X}}$, and $u,\tilde{u} \in {\mathcal{U}}$.
\end{definition}
Note that the existence of an i-IOSS Lyapunov function is sufficient for the system to be i-IOSS \cite{rawlings2017model}, the notion of detectability of interest in this work.

Under Ass.~\ref{assump:output-affine}, the DT LMI conditions in \eqref{eqn:prop:lyapunov-function-lmi-condition} can be used to establish an i-IOSS Lyapunov function for an ADT system \eqref{eqn:approx-system-dynamics}-\eqref{eqn:approx-system-output-equation} with a \textit{specific} sampling period, as shown in Prop.~\ref{prop:lyapunov-function}.

\begin{assumption} \label{assump:output-affine}
    The output function $h$ \eqref{eqn:output-equation} is affine in $(x,u)$.
\end{assumption}

\begin{proposition}[{i-IOSS Lyapunov function via DT LMI conditions \cite[Cor.~3]{schiller2023lyapunov}}] \label{prop:lyapunov-function}
    Consider a convex set $\mathcal{X}\times\mathcal{U}\times\mathcal{D}\subseteq\mathbb{R}^n\times\mathbb{R}^q\times\mathbb{R}^m$, and an ADT system \eqref{eqn:approx-system-dynamics}-\eqref{eqn:approx-system-output-equation} under Ass.~\ref{assump:h-cont-diff}-\ref{assump:output-affine}, with specific sampling period $\tau > 0$. Suppose it satisfies the DT LMI conditions in the sense of Def.~\ref{def:dt-lmi} on $\mathcal{X}\times\mathcal{U}\times\mathcal{D}$ for some $P \in \mathbb{S}_{++}^n$, $Q \in \mathbb{S}_{+}^q$, $R \in \mathbb{S}_{+}^p$, and $\eta \in [0,1)$. Then, $W(x,\tilde{x}) = \Norm{x - \tilde{x}}_P^2$ is an i-IOSS Lyapunov function on $\mathcal{X}\times\mathcal{U}\times\mathcal{D}$ in the sense of Def.~\ref{def:dt-lyapunov}.
\end{proposition}

Finally, we provide Cor.~\ref{cor:ct-lmi-implies-lyap}. It says that under the premise of Thm.~\ref{thm:ct-implies-dt} and additionally Ass.~\ref{assump:output-affine},
an i-IOSS Lyapunov function for a \textit{family} of ADT system can be established. 

\begin{corollary} \label{cor:ct-lmi-implies-lyap}
    Consider a convex set $\mathcal{X}\times\mathcal{U}\times\mathcal{D}\subseteq\mathbb{R}^n\times\mathbb{R}^q\times\mathbb{R}^m$, and the CT system \eqref{eqn:system-dynamics}-\eqref{eqn:output-equation} and the family of ADT systems \eqref{eqn:approx-system-dynamics}-\eqref{eqn:approx-system-output-equation} under Ass.~\ref{assump:ct-cont-diff}-\ref{assump:output-affine}. Suppose conditions~1-2 in Thm.~\ref{thm:ct-implies-dt} are satisfied. Then, for each ADT system \eqref{eqn:approx-system-dynamics}-\eqref{eqn:approx-system-output-equation} with sampling period $\tau \in (0,\tau_1)$, 
    $W(x,\tilde{x}) = \Norm{x - \tilde{x}}_P^2$ is an i-IOSS Lyapunov function on $\mathcal{X}\times\mathcal{U}\times\mathcal{D}$, with supply matrices $\tilde{Q}(\tau) \in \mathbb{S}_{++}^q$ and $\tilde{R}(\tau) \in \mathbb{S}_{++}^p$, and factor $\eta(\tau) \in (0,1)$.
    Recall that $\tau_1$, $\tilde{Q}$, $\tilde{R}$, and $\eta$, are from \eqref{eqn:theorem:ct-implies-dt-tau1}-\eqref{eqn:theorem:ct-implies-dt-eta}.
\end{corollary}
\begin{remark}
    The Lyapunov function $W(x,\tilde{x})$, supply matrices $\tilde{Q}(\tau)$ and $\tilde{R}(\tau)$, factor $\eta(\tau)$, and sampling period bound $\tau_1$, in Cor.~\ref{cor:ct-lmi-implies-lyap}, can be constructed given explicit knowledge of $L_f$, $\tau_0$, and $\rho$ (recall Rem.~\ref{rem:explicit}).
\end{remark}

\subsection{Consistency of linearizations for Euler- and RK2-discretized models} \label{sec:consistency-euler-rk2}
In Sec.~\ref{sec:detectability-class}, the discretization method needed to be chosen so that the family of linearizations for the ADT system are consistent with the CT system linearizations; i.e., we provided a result that holds for the class of discretization schemes yielding consistent linearizations. We will now explicitly show that the Euler method always satisfies this requirement, and that the RK2 method also satisfies this under (i) a Lipschitz assumption on the CT system, and (ii) a different assumption on the linearizations of the CT system.

The family of ADT models obtained via the Euler method is defined using the CT system model as
\begin{align}
    F^{\textnormal{Euler}}_{\tau}(x,u,d) := x + \tau f(x,u,d) \label{eqn:euler-model}
\end{align}
for $\tau > 0$.
The consistency of the linearizations of these dynamics is explicitly characterized in Lem.~\ref{lemma:euler}.
\begin{lemma} \label{lemma:euler}
    Consider the CT system \eqref{eqn:system-dynamics}-\eqref{eqn:output-equation} and the family of ADT systems \eqref{eqn:approx-system-dynamics}-\eqref{eqn:approx-system-output-equation} corresponding to Euler discretization \eqref{eqn:euler-model}, such that $F_{\tau} = F^{\textnormal{Euler}}_{\tau}$ for $\tau > 0$. Suppose Ass.~\ref{assump:ct-cont-diff} holds. 
    Then, Ass.~\ref{assump:dt-diff} is satisfied, and moreover, for any $\rho \in \mathcal{K}$ and $\tau_0 > 0$, $\{(\tilde{A}_{\tau},\tilde{B}_{\tau})\}_{\tau > 0}$ \eqref{eqn:dt-linearized-A}-\eqref{eqn:dt-linearized-B} is consistent with $(A,B)$ \eqref{eqn:ct-linearized-AB} on $\mathbb{R}^n\times\mathbb{R}^q\times\mathbb{R}^m$ in the sense of Def.~\ref{def:consistent-lineariztions}, i.e., condition~2 of Thm.~\ref{thm:ct-implies-dt} is satisfied.
\end{lemma}

On the other hand, the family of models corresponding to the second-order RK (RK2) method is defined as
\begin{align}
    F^{\textnormal{RK2}}_{\tau}(x,u,d):= x + \tau f\Parenth{ x + \frac{\tau}{2} f\Parenth{x,u,d}, u,d} \label{eqn:rk2-model}
\end{align}
for $\tau > 0$.
The consistency of the linearizations of these dynamics is explicitly characterized in Lem.~\ref{lemma:rk2}.
\begin{lemma} \label{lemma:rk2}
    Consider a set $\mathcal{X}\times\mathcal{U}\times\mathcal{D}\subseteq\mathbb{R}^n\times\mathbb{R}^q\times\mathbb{R}^m$, the CT system \eqref{eqn:system-dynamics}-\eqref{eqn:output-equation}, and the family of ADT systems \eqref{eqn:approx-system-dynamics}-\eqref{eqn:approx-system-output-equation} corresponding to RK2 discretization \eqref{eqn:rk2-model}, such that $F_{\tau} = F^{\textnormal{RK2}}_{\tau}$ for $\tau > 0$. 
    Suppose the following conditions hold:
    \begin{enumerate}
        \item Ass.~\ref{assump:ct-cont-diff} holds and $f$ in \eqref{eqn:system-dynamics} is $L_f$-Lipschitz on $\mathcal{X}\times\mathcal{U}\times\mathcal{D}$;
        \item There exist $\delta_0 > 0$ and $\sigma \in \mathcal{K}$ such that for all $d \in {\mathcal{D}}$, $u \in {\mathcal{U}}$, and $x,\tilde{x} \in {\mathcal{X}}$ satisfying $\Norm{x-\tilde{x}}_2 \leq \Norm{f(\tilde{x},u,d)}_2 \delta_0$ and $\Norm{f(\tilde{x},u,d)}_2 > 0$,
    \begin{align*}
            &\max \big \{ \Norm{A(x,u,d) - A(\tilde{x},u,d)}_2, \\
            & \quad \quad \ \ \Norm{B(x,u,d) - B(\tilde{x},u,d)}_2   \big \} \\
            &\leq \sigma \Parenth{ \frac{\Norm{x - \tilde{x}}_2}{\Norm{f(\tilde{x},u,d)}_2} }. 
    \end{align*}
    \end{enumerate}
    Then, Ass.~\ref{assump:dt-diff} is satisfied, and $\{(\tilde{A}_{\tau},\tilde{B}_{\tau})\}_{\tau > 0}$ \eqref{eqn:dt-linearized-A}-\eqref{eqn:dt-linearized-B} is consistent with $(A,B)$ \eqref{eqn:ct-linearized-AB} on $\mathcal{X}\times\mathcal{U}\times\mathcal{D}$ in the sense of Def.~\ref{def:consistent-lineariztions}, i.e., condition~2 of Thm.~\ref{thm:ct-implies-dt} is satisfied, with
    \begin{align}
        \tau_0 = 2 \delta_0, \quad \rho(\tau) = \sigma\Parenth{ \frac{\tau}{2} } + \frac{\tau}{2}L_f^2.
    \end{align}
\end{lemma}

\begin{remark} \label{rem:rk2}
    Intuitively, condition~2 of Lem.~\ref{lemma:rk2} says that when $({x},u,d) \in \mathcal{X}\times\mathcal{U}\times\mathcal{D}$ is such that the magnitude of $f({x},u,d)$ is large, the slope of the linearizations (w.r.t. the first argument) at $({x},u,d)$ is less steep. 
    This 
    holds
    in the case of linear system dynamics globally, since $A(x,u,d)=A(\tilde{x},u,d)$ and $B(x,u,d)=B(\tilde{x},u,d)$ for all $d \in \mathbb{R}^m$, $u \in \mathbb{R}^q$, and $x,\tilde{x}\in\mathbb{R}^n$.
    However, it is not implied by Lipschitz continuity of $f$. 
    For example, consider
    $f(x) = x + \sin(x)$, where $A(x) = \deriv{x} f(x) = 1 + \cos(x)$ is uniformly bounded in $x \in \mathbb{R}$ and therefore $f$ is Lipschitz continuous. Although $\lim_{x \rightarrow \infty}f(x) = +\infty$, the slope of $A(x)$ (given by $\deriv{x} A(x) = -\sin(x)$) keeps oscillating at the same frequency as $x \rightarrow \infty$, and hence does not satisfy condition~2.
\end{remark}

Rem.~\ref{rem:rk2} shows that condition~2 of Lem.~\ref{lemma:rk2} can be restrictive in general. However, this is satisfied when ${\mathcal{X}}$, ${\mathcal{U}}$, and ${\mathcal{D}}$ are compact, as shown in Lem.~\ref{lemma:compact2}. This is practically relevant when gridding is used to verify the LMIs (recall Rem.~\ref{rem:finite}).

\begin{lemma} \label{lemma:compact2}
    Consider a set $\mathcal{X}\times\mathcal{U}\times\mathcal{D}$, and the CT system \eqref{eqn:system-dynamics}-\eqref{eqn:output-equation}. 
    Suppose $f$ is twice continuously differentiable, and that ${\mathcal{X}}$, ${\mathcal{D}}$, and ${\mathcal{U}}$, are compact. Then, 
    condition~2 in Lem.~\ref{lemma:rk2} is satisfied for any $\delta_0 > 0$ with $\sigma(s) = L_{df} \cdot c_f \cdot s$, where $c_f=\max_{(x,u,d)\in{\mathcal{X}}\times{\mathcal{U}}\times{\mathcal{D}}}\Norm{f(x,u,d)}_2$, and $L_{df}>0$ is the Lipschitz constant of $\begin{bmatrix}
        A & B
    \end{bmatrix}$ on $\mathcal{X}\times\mathcal{U}\times\mathcal{D}$. 
\end{lemma}

%% file: 04-examples.tex
\section{Example}

We now
provide an example showing that it takes less computational effort to verify the CT LMI conditions from Def.~\ref{def:ct-lmi} than the DT ones from Def.~\ref{def:dt-lmi} on the ADT system, demonstrating the advantages of our results in Sec.~\ref{sec:results}.

Consider the following system:
\begin{equation} 
\begin{aligned} 
    \dot{x} = f(x,u) &= \begin{bmatrix}
        -2 k_1 x_1^2 + 2k_2 x_2 + u_1 \\
        k_1 x_1^2 - k_2 x_2 + u_2
    \end{bmatrix}, \\
    y &= x_1 + x_2 + u_3,
\end{aligned} \label{eqn:reactor-system}
\end{equation}
with $k_1 = 0.16$ and $k_2 = 0.0064$. Here, $x = \begin{bmatrix} x_1 & x_2 \end{bmatrix}^{\top} \in \mathbb{R}^2$ is the state, $u = \begin{bmatrix}u_1 & u_2 & u_3 \end{bmatrix}^{\top}$ is the disturbance (system input in \eqref{eqn:system-dynamics}), and $y \in \mathbb{R}$ is the output. It corresponds to a reversible chemical reaction taking place in an isothermal gas-phase reactor (see \cite{tenny2002efficient}). Suppose that RK2 discretization via \eqref{eqn:rk2-model} has been used to obtain an ADT system model.

To see that it is easier to verify the CT LMI conditions than the DT ones on the ADT systems, we will compare the linearizations of the respective systems.
Firstly, note that the linearizations of the CT dynamics \eqref{eqn:reactor-system} are
\begin{align}
    &A(x)= \begin{bmatrix}
        -4k_1x_1 & 2 k_2 \\ 2k_1x_1 & -k_2
    \end{bmatrix}, \quad && B = \begin{bmatrix}
        I_2 & 0_{2 \times 1}
    \end{bmatrix}. \label{eqn:reactor-ct-linearization-1}
\end{align}
In contrast, the linearizations of the RK2 dynamics are
\begin{align*}
    &\tilde{A}_{\tau}(x,u) = \deriv{x} \Brack{x + \tau f\Parenth{x + \frac{\tau}{2}f(x,u),u}} \\
    &= I + \tau \begin{bmatrix}
        -4k_1 \Parenth{x_1 + \frac{\tau}{2} \Parenth{ -2 k_1 x_1^2 + 2 k_2 x_2 + u_1}} & 2 k_2 \\ 2k_1 \Parenth{x_1 + \frac{\tau}{2} \Parenth{ -2 k_1 x_1^2 + 2 k_2 x_2 + u_1}} & -k_2
    \end{bmatrix} \\
    & \quad \cdot \Parenth{I + \frac{\tau}{2} \begin{bmatrix}
        -4k_1x_1 & 2 k_2 \\ 2k_1x_1 & -k_2
    \end{bmatrix} } 
\end{align*}
and
\begin{align*}
    &\tilde{B}_{\tau}(x,u) = \deriv{u} \Brack{x + \tau f\Parenth{x + \frac{\tau}{2}f(x,u),u}} \\
    &= \frac{\tau^2}{2} \begin{bmatrix}
        -4k_1 \Parenth{x_1 + \frac{\tau}{2} \Parenth{ -2 k_1 x_1^2 + 2 k_2 x_2 + u_1}} & 2 k_2 & 0 \\ 2k_1 \Parenth{x_1 + \frac{\tau}{2} \Parenth{ -2 k_1 x_1^2 + 2 k_2 x_2 + u_1}} & -k_2 & 0
    \end{bmatrix} \\
    & \quad + \tau \begin{bmatrix}
        I_2 & 0_{2 \times 1}
    \end{bmatrix} \label{eqn:reactor-dt-linearization-2}
\end{align*} 
for $\tau > 0$.
The CT system linearizations are clearly simpler, and therefore easier to compute than the ADT system.
Moreover, 
verification of the LMI conditions for the CT system can be performed over a lower-dimensional space compared to the ADT system since ($\tilde{A}_{\tau}$, $\tilde{B}_{\tau}$) depend on $(x,u)$, whereas ($A$,$B$) depend only on $x$. This is significant when verifying the LMIs via gridding (recall Rem.~\ref{rem:finite}), since the computation time suffers from the curse of dimensionality. 

Numerical experiments were also performed. The linearizations of the CT and ADT models were computed over the set $\mathcal{X}\times\mathcal{U} = [0.1,0.5]^2\times[-0.1,0.1]^3$ on a laptop with an Intel Core i7-13850HX CPU running MATLAB, with 100 grid points in every direction. Computing the linearizations took $1.27\times10^{-4}$~s for the CT system, versus $1.16$~s for the ADT system. We exploited the fact that the CT linearizations depend only on $x_1$, and the ADT linearizations depend on $x_1$, $x_2$, and $u_1$. This demonstrates the benefit of our approach.

%% file: 05-conclusion.tex
\section{Conclusion}
In this work, we investigated the problem of obtaining checkable conditions on CT systems, that can be used to verify i-IOSS of families of ADT systems parameterized by sampling period and obtained via some numerical discretization method.
Under a Lipschitz assumption on the CT system,
we showed that the satisfaction of LMIs for the CT system enables verification of analogous infinite-dimensional LMIs for the corresponding family of ADT system. This is achieved supposing that the linearizations of the ADT system are consistent with the linearizations of the CT system, and the sampling period is sufficiently small. We then make connections with \cite{schiller2023lyapunov} to show that this is sufficient for the construction of i-IOSS Lyapunov functions for the family of ADT systems. 
Subsequently, we show that discretizaton via the Euler and RK2 methods produces families of ADT systems with linearizations that are consistent with those of the CT system, under regularity conditions on the CT system.
Lastly, a simple example is provided showing that it is more computationally efficient to verify finite-dimensional approximations of the CT LMIs than the DT ones on the ADT system, further highlighting the utility of our result. 

In future work, it would be interesting to investigate how the consistency of linearizations in Def.~\ref{def:consistent-lineariztions} is related to standard notions of consistency for approximately discretized systems (see \cite{nesic1999sufficient,nesic2004framework,arcak2004framework,Burden2010}), and generalize our results beyond i-IOSS to incremental dissipation inequalities.

%% file: 05-proofs.tex
\begin{proof}[Proof of Thm.~\ref{thm:ct-implies-dt}]
    In the first part of this proof, we will establish that for all $(x,u,d)\in {\mathcal{X}}\times{\mathcal{U}}\times{\mathcal{D}}$ and $\tau \in (0,\tau_0]$,
    \begin{align}
        &\begin{bmatrix}
            \tilde{A}_{\tau}^{\top} P \tilde{A}_{\tau}  & \tilde{A}_{\tau}^{\top} P \tilde{B}_{\tau}  \\
            \tilde{B}_{\tau}^{\top} P \tilde{A}_{\tau} & \tilde{B}_{\tau}^{\top} P \tilde{B}_{\tau}
        \end{bmatrix} \preceq \eta(\tau) \begin{bmatrix}
            P & 0 \\ 0 & 0
        \end{bmatrix} \\
        & \quad + \tau \begin{bmatrix}
            C^{\top} R C & C^{\top} R D \\ D^{\top} R C & D^{\top} R D + {Q + \alpha(\tau) \lambdamin(P) I_q}
        \end{bmatrix}. \quad \label{eqn:theorem:ct-implies-dt-11}
    \end{align}
    Let $r_{\tau}(x,u,d) = F_{\tau}(x,u,d) - x - \tau f(x,u,d)$ denote the error between the ADT system dynamics $F_{\tau}(x,u,d)$ and the Euler dynamics $x + \tau f(x,u,d)$. 
    For the sake of brevity, denote $F_{\tau}(x,u,d)$, $f(x,u,d)$ and $r_{\tau}(x,u,d)$, by $F_{\tau}$, $f$, and $r_{\tau}$ respectively.
    Then, for all $(x,u,d)\in {\mathcal{X}}\times{\mathcal{U}}\times{\mathcal{D}}$ and $\tau > 0$, the linearizations \eqref{eqn:dt-linearized-A} and \eqref{eqn:dt-linearized-B} can be rewritten as
    \begin{align}
        \tilde{A}_{\tau} &= \deriv{x} F_{\tau} = \deriv{x}\Brack{x + \tau f + r_{\tau}} = I_n + \tau A + \deriv{x} r_{\tau}, \label{eqn:theorem:ct-implies-dt-1} \quad \\
        \tilde{B}_{\tau} &= \deriv{u} F_{\tau}= \deriv{u}\Brack{x + \tau f + r_{\tau}}  =\tau B + \deriv{u} r_{\tau}, \label{eqn:theorem:ct-implies-dt-2}
    \end{align}
    by applying \eqref{eqn:ct-linearized-AB}. Using this, it follows that for all $(x,u,d)\in {\mathcal{X}}\times{\mathcal{U}}\times{\mathcal{D}}$ and $\tau > 0$, the LHS of \eqref{eqn:theorem:ct-implies-dt-11} can be rewritten as
    \begin{align}
        &\begin{bmatrix}
            \tilde{A}_{\tau}^{\top} P \tilde{A}_{\tau}  & \tilde{A}_{\tau}^{\top} P \tilde{B}_{\tau}  \\
            \tilde{B}_{\tau}^{\top} P \tilde{A}_{\tau} & \tilde{B}_{\tau}^{\top} P \tilde{B}_{\tau}
        \end{bmatrix} = \begin{bmatrix}
            \tilde{A}_{\tau} & \tilde{B}_{\tau}
        \end{bmatrix}^{\top} P \begin{bmatrix}
            \tilde{A}_{\tau} & \tilde{B}_{\tau}
        \end{bmatrix} \\
        & = \begin{bmatrix}
            \Parenth{I_n + \tau A }^{\top} \\
            \Parenth{\tau B }^{\top}
        \end{bmatrix} P  \begin{bmatrix}
            I_n + \tau A  &
            \tau B 
        \end{bmatrix} \label{eqn:theorem:ct-implies-dt-4} \\
        & \quad + \begin{bmatrix}
            \Parenth{I_n + \tau A }^{\top} \\
            \Parenth{\tau B }^{\top}
        \end{bmatrix} P  \begin{bmatrix}
            {\deriv{x} r_{\tau}} &
            {\deriv{u} r_{\tau}}
        \end{bmatrix} \label{eqn:theorem:ct-implies-dt-5} \\
        & \quad + \begin{bmatrix}
            \Parenth{\deriv{x} r_{\tau}}^{\top} \\
            \Parenth{\deriv{u} r_{\tau}}^{\top}
        \end{bmatrix} P \begin{bmatrix}
            {I_n + \tau A } &
            {\tau B }
        \end{bmatrix} \label{eqn:theorem:ct-implies-dt-6} \\
        & \quad + \begin{bmatrix}
            \Parenth{\deriv{x} r_{\tau}}^{\top} \\
            \Parenth{\deriv{u} r_{\tau}}^{\top}
        \end{bmatrix} P \begin{bmatrix}
            {\deriv{x} r_{\tau}} &
            {\deriv{u} r_{\tau}}
        \end{bmatrix}. \label{eqn:theorem:ct-implies-dt-3}
    \end{align}
    
    We now bound the terms in \eqref{eqn:theorem:ct-implies-dt-4}-\eqref{eqn:theorem:ct-implies-dt-3}. Firstly, \eqref{eqn:theorem:ct-implies-dt-4} is bounded as follows for all $(x,u,d)\in {\mathcal{X}}\times{\mathcal{U}}\times{\mathcal{D}}$ and $\tau > 0$:
    \begin{align}
        &\begin{bmatrix}
            \Parenth{I_n + \tau A }^{\top} \\
            \Parenth{\tau B }^{\top}
        \end{bmatrix} P  \begin{bmatrix}
            {I_n + \tau A } &
            {\tau B }
        \end{bmatrix} \preceq \tau^2 L_f^2 \lambdamax(P) I_{n+q} \\
        & \quad + \Parenth{1-\tau \kappa}\begin{bmatrix}
            P & 0 \\ 0 & 0
        \end{bmatrix} + \tau \begin{bmatrix}
            C^{\top} R C & C^{\top} R D \\ D^{\top} R C & D^{\top} R D + Q
        \end{bmatrix} \label{eqn:theorem:ct-implies-dt-7}
    \end{align}
    where \eqref{eqn:theorem:ct-implies-dt-7} follows since $f$ is $L_f$-Lipschitz, and by using condition 3 in the premise of Thm.~\ref{thm:ct-implies-dt} to apply \eqref{eqn:thm:ct-lmi-implies-dt-lmi-diff-quot-ct-condition}.
    
    Secondly, \eqref{eqn:theorem:ct-implies-dt-5}-\eqref{eqn:theorem:ct-implies-dt-6} can be upper bounded as follows for all $(x,u,d)\in {\mathcal{X}}\times{\mathcal{U}}\times{\mathcal{D}}$ and $\tau \in (0,\tau_0]$:
    \begin{align}
        &\begin{bmatrix}
            \Parenth{I_n + \tau A }^{\top} \\
            \Parenth{\tau B }^{\top}
        \end{bmatrix} P  \begin{bmatrix}
            \Parenth{\deriv{x} r_{\tau}} &
            \Parenth{\deriv{u} r_{\tau}}
        \end{bmatrix} \\
        &\quad + \begin{bmatrix}
            \Parenth{\deriv{x} r_{\tau}}^{\top} \\
            \Parenth{\deriv{u} r_{\tau}}^{\top}
        \end{bmatrix} P \begin{bmatrix}
            \Parenth{I_n + \tau A } &
            \Parenth{\tau B }
        \end{bmatrix} \\
        &\preceq 4\Parenth{1 + \tau L_f } \lambda_{\textnormal{max}}\Parenth{P} \tau \rho\Parenth{\tau} I_{n+q}. \label{eqn:theorem:ct-implies-dt-8}
    \end{align}
    Here, \eqref{eqn:theorem:ct-implies-dt-8} follows since $f$ is $L_f$-Lipschitz, 
    and also from
    \begin{align}
        &\max \Brace{\Norm{\deriv{x} r_{\tau}}_2, \Norm{\deriv{u} r_{\tau}}_2 } \\
        &=\tau \max \Brace{\Norm{\frac{\tilde{A}_{\tau}-I_n}{\tau} - A}_2, \Norm{\frac{\tilde{B}_{\tau}}{\tau} - B}_2 }
        \leq \tau \rho(\tau) \quad \  \label{eqn:theorem:ct-implies-dt-9}
    \end{align}
    for all $(x,u,d)\in {\mathcal{X}}\times{\mathcal{U}}\times{\mathcal{D}}$ and $\tau \in (0,\tau_0]$, where the equality in \eqref{eqn:theorem:ct-implies-dt-9} follows from \eqref{eqn:theorem:ct-implies-dt-1} and \eqref{eqn:theorem:ct-implies-dt-2}, and the inequality holds by the consistency of linearizations in condition~2 of Thm.~\ref{thm:ct-implies-dt}.
    Thirdly, \eqref{eqn:theorem:ct-implies-dt-3} can be upper bounded as follows for all $(x,u,d)\in {\mathcal{X}}\times{\mathcal{U}}\times{\mathcal{D}}$ and $\tau \in (0,\tau_0]$:    
    \begin{align}
        &\begin{bmatrix}
            \Parenth{\deriv{x} r_{\tau}}^{\top} \\
            \Parenth{\deriv{u} r_{\tau}}^{\top}
        \end{bmatrix} P \begin{bmatrix}
            {\deriv{x} r_{\tau}} &
            {\deriv{u} r_{\tau}}
        \end{bmatrix} 
        \preceq 4 \lambdamax(P) \tau^2 \rho(\tau)^2 I_{n+q} \quad \ \label{eqn:theorem:ct-implies-dt-10}
    \end{align}
    where \eqref{eqn:theorem:ct-implies-dt-10} also makes use of \eqref{eqn:theorem:ct-implies-dt-9}.
    
    By combining \eqref{eqn:theorem:ct-implies-dt-4}-\eqref{eqn:theorem:ct-implies-dt-3}, \eqref{eqn:theorem:ct-implies-dt-7}, \eqref{eqn:theorem:ct-implies-dt-8}, and \eqref{eqn:theorem:ct-implies-dt-10}, and using the definition of $\alpha$ from  \eqref{eqn:theorem:ct-implies-dt-alpha1}, we find that for all $\tau \in (0,\tau_0]$ and $(x,u,d)\in {\mathcal{X}}\times{\mathcal{U}}\times{\mathcal{D}}$,
    \begin{align}
        &\begin{bmatrix}
            \tilde{A}_{\tau}^{\top} P \tilde{A}_{\tau}  & \tilde{A}_{\tau}^{\top} P \tilde{B}_{\tau}  \\
            \tilde{B}_{\tau}^{\top} P \tilde{A}_{\tau} & \tilde{B}_{\tau}^{\top} P \tilde{B}_{\tau}
        \end{bmatrix} \preceq \Parenth{1-\tau \kappa}\begin{bmatrix}
            P & 0 \\ 0 & 0
        \end{bmatrix} \label{eqn:theorem:ct-implies-dt-12} \\
        &\quad + \tau \begin{bmatrix}
            C^{\top} R C & C^{\top} R D \\ D^{\top} R C & D^{\top} R D + Q
        \end{bmatrix}  + \tau \alpha(\tau) \lambdamin(P) I_{n+q}
   \end{align}
   Next, using the definition of $\eta$ from \eqref{eqn:theorem:ct-implies-dt-eta}, we know that for all $\tau \in (0,\tau_0]$ and $(x,u,d)\in {\mathcal{X}}\times{\mathcal{U}}\times{\mathcal{D}}$,
    $
        \tau \alpha(\tau) \lambdamin(P) I_n =  \Parenth{ \eta (\tau) - 1 + \kappa \tau } \lambdamin(P) I_n 
        \preceq  \Parenth{ \eta (\tau) - 1 + \kappa \tau } P,
    $
    which can subsequently be rearranged to obtain
    \begin{align}
        \Parenth{1 - \tau \kappa }P + \tau \alpha(\tau) \lambdamin(P) I_n &\preceq  \eta (\tau) P. \label{eqn:theorem:ct-implies-dt-13}
    \end{align}
    Then, we finally arrive at \eqref{eqn:theorem:ct-implies-dt-11} by combining \eqref{eqn:theorem:ct-implies-dt-12} with \eqref{eqn:theorem:ct-implies-dt-13}. 
    concluding the first part of this proof.

    In the second part of this proof, we establish that $\eta(\tau) \in (0,1)$ for all $\tau \in (0,\tau_1)$, with $\tau_1$ defined in \eqref{eqn:theorem:ct-implies-dt-tau1}.
    Note that by selecting $\tau \in \Parenth{ 0, \alpha^{-1} \Parenth{ \kappa } }$, $\tau \Parenth{ \alpha(\tau) - \kappa } < 0$ follows.
    Moreover, after selecting $\tau \in \Parenth{0, \frac{1}{\kappa}}$, $-1 < - \tau \kappa < \tau (\alpha(\tau) - \kappa )$ holds.
    It follows that for all $\tau \in \Parenth{0, \tau_1} = \Parenth{ 0, \min \Brace{ \frac{1}{\kappa} , \tau_0, \alpha^{-1}(\kappa) } }$, $-1 < \tau \Parenth{\alpha(\tau) - \kappa} < 0$ holds, which is equivalent to 
    \begin{align}
        0 < \eta(\tau) = \tau \alpha(\tau) + 1 - \kappa \tau < 1, \label{eqn:theorem:ct-implies-dt-14}
    \end{align}
    concluding the second part of the proof.

    We reach
    the conclusion 
    since 
    by
    \eqref{eqn:theorem:ct-implies-dt-11} and \eqref{eqn:theorem:ct-implies-dt-14}, \eqref{eqn:prop:lyapunov-function-lmi-condition} 
    holds
    for all $(x,u,d)\in{\mathcal{X}}\times{\mathcal{U}}\times{\mathcal{D}}$ and $\tau \in (0,\tau_1)$ with Lyapunov matrix $P\in \mathbb{S}_{++}^n$, supply matrices $\tilde{Q}(\tau) \in \mathbb{S}_{++}^q$ and $\tilde{R}(\tau) \in \mathbb{S}_{++}^p$, and factor $\eta(\tau) \in (0,1)$
    (defined 
    in \eqref{eqn:theorem:ct-implies-dt-tau1}-\eqref{eqn:theorem:ct-implies-dt-eta}).
\end{proof}

\begin{proof}[Proof of Cor.~\ref{cor:ct-lmi-implies-lyap}]
    Under the premise of Cor.~\ref{cor:ct-lmi-implies-lyap}, the premise of Thm.~\ref{thm:ct-implies-dt} is satisfied, and so 
    each ADT system \eqref{eqn:approx-system-dynamics}-\eqref{eqn:approx-system-output-equation} with sampling period $\tau \in \Parenth{0,\tau_1}$ satisfies the DT LMI conditions in Def.~\ref{def:dt-lmi} on $\mathcal{X}\times\mathcal{U}\times\mathcal{D}$ with Lyapunov matrix $P\in \mathbb{S}_{++}^n$, supply matrices $\tilde{Q}(\tau) \in \mathbb{S}_{++}^q$ and $\tilde{R}(\tau) \in \mathbb{S}_{++}^p$, and factor $\eta(\tau) \in (0,1)$.
    Combined with Ass.~\ref{assump:h-cont-diff}-\ref{assump:output-affine}, 
    this 
    satisfies
    the premise of Prop.~\ref{prop:lyapunov-function}, yielding the desired conclusion.
\end{proof}

\begin{proof}[Proof of Lem.~\ref{lemma:euler}]
    Firstly, let $F = F^{\textnormal{Euler}}$, such that $\tilde{A}_{\tau} = \deriv{x} F^{\textnormal{Euler}}_{\tau}$ and $\tilde{B}_{\tau} = \deriv{u} F^{\textnormal{Euler}}_{\tau}$.
    Next, note that Ass.~\ref{assump:dt-diff} is automatically satisfied by construction using Ass.~\ref{assump:ct-cont-diff}.
    Now, note that for all $(x,u,d)\in \mathbb{R}^n\times\mathbb{R}^q\times\mathbb{R}^m$ and $\tau > 0$, 
    $
        \lVert (\tilde{A}_{\tau}-I)/\tau - A \rVert_2 = \Norm{((I_n+\tau A)-I_n)/\tau - A}_2 = 0
    $ and $
        \lVert \tilde{B}_{\tau}/\tau - B \rVert_2 = \Norm{(\tau B)/\tau - B}_2 = 0
    $.
    Then, it trivially follows that for any choice of $\rho \in \mathcal{K}$ and $\tau > 0$, $\max \{ \lVert  (\tilde{A}_{\tau}-I)/\tau - A \rVert_2, \lVert \tilde{B}_{\tau}/\tau - B \rVert_2 \} = 0 \leq \rho(\tau)$ for all $(x,u,d)\in\mathbb{R}^n\times\mathbb{R}^q\times\mathbb{R}^m$, yielding the conclusion.    
\end{proof}

\begin{proof}[Proof of Lem.~\ref{lemma:rk2}]
    For the sake of brevity, throughout this proof, we denote $f(x,u,d)$ by $f$.
    Firstly, let $F = F^{\textnormal{RK2}}$, such that $\tilde{A}_{\tau} = \deriv{x} F^{\textnormal{RK2}}_{\tau}$ and $\tilde{B}_{\tau} = \deriv{u} F^{\textnormal{RK2}}_{\tau}$. 
    Next, note that Ass.~\ref{assump:dt-diff} is automatically satisfied by construction using Ass.~\ref{assump:ct-cont-diff}.
    Now note that for all $(x,u,d)\in {\mathcal{X}}\times{\mathcal{U}}\times{\mathcal{D}}$ and $\tau \leq 2 \delta_0$, 
    $
        \tilde{A}_{\tau} = \deriv{x} [ x + \tau f ( x + \tau f/2 ,u,d) ] 
        = I_n + \tau A ( x + \tau f/2,u,d ) \Brack{I_n + \tau A / 2}
    $
    and therefore
    \begin{align}
        &\Norm{\frac{\tilde{A}_{\tau}-I_n}{\tau} - A}_2 \\
        &
        =
        \Norm{A\Parenth{x + \frac{\tau}{2}f,u,d} - A + \frac{\tau}{2} A \Parenth{x + \frac{\tau}{2}f,u,d} A}_2 \\
        &\leq \sigma\Parenth{ \frac{\tau}{2} } + \frac{\tau}{2}L_f^2. \label{eqn:lemma:rk2-1}
    \end{align}
    Here, \eqref{eqn:lemma:rk2-1} follows since $f$ is $L_f$-Lipschitz, and the fact that when $\tau \leq 2\delta_0$, $\Norm{\Parenth{x + \tau f / 2}-x}_2 \leq \Norm{f}_2 \delta_0$ is satisfied, such that condition 2 in Lem.~\ref{lemma:rk2} can be applied to show that
    $
        \Norm{A\Parenth{x + \tau f/2,u,d} - A}_2 \leq 
        \sigma \Parenth{\tau/2}
    $.
    Similarly, for all $(x,u,d)\in {\mathcal{X}}\times{\mathcal{U}}\times{\mathcal{D}}$ and $\tau \leq 2 \delta_0$, 
    $
        \tilde{B}_{\tau} = \deriv{u} [ x + \tau f\Parenth{ x + \tau f/2  ,u,d} ] 
        = \tau A ( x + \tau f / 2,u,d ) [\tau B / 2] + \tau B (x + \tau f / 2,u,d)
        = \tau A (x + \tau f/2,u,d) B / 2 + B (x + \tau f / 2,u,d)
    $,
    and therefore
    $
        \lVert \tilde{B}_{\tau}/\tau - B \rVert_2 
        = \lVert \tau A\Parenth{x + \tau f / 2,u,d} B / 2 + B\Parenth{x + \tau / 2,u,d} - B\rVert_2 
        \leq \sigma\Parenth{\tau/2} + \tau L_f^2/2
    $.
    The conclusion immediately follows by combining this with \eqref{eqn:lemma:rk2-1}.
\end{proof}

\begin{proof}[Proof of Lem.~\ref{lemma:compact2}]
    Note that $L_{df}$ from Lem.~\ref{lemma:compact2} exists since $f$ is twice continuously differentiable, so $\begin{bmatrix}
        A & B
    \end{bmatrix}$ is continuously differentiable, and therefore Lipschitz continuous on a compact set. 
    It follows that for all $d \in {\mathcal{D}}$, $u \in \mathcal{U}$, and $x,\tilde{x} \in \mathcal{X}$, if $\Norm{f(\tilde{x},u,d)}_2 > 0$, then
    \begin{align}
        \Norm{A(x,u,d)-A(\tilde{x},u,d)}_2 &\leq L_{df} \Norm{x - \tilde{x}}_2 \label{eqn:lemma:compact-1}\\
        &\leq L_{df} c_f \frac{\Norm{x - \tilde{x}}_2}{\Norm{f(\tilde{x},u,d)}_2}, \label{eqn:lemma:compact-3}
    \end{align}
    where \eqref{eqn:lemma:compact-3} holds by making use of the continuity of $f$ from Ass.~\ref{assump:ct-cont-diff} and the compactness of $\mathcal{X}\times\mathcal{U}\times\mathcal{D}$, such that $c_f:=\max_{(x,u,d)\in{\mathcal{X}}\times{\mathcal{U}}\times{\mathcal{D}}}\Norm{f(x,u,d)}_2 < \infty$. Following similar steps, we can also establish that $\Norm{B(x,u,d)-B(\tilde{x},u,d)}_2 \leq L_{df} c_f \frac{\Norm{x - \tilde{x}}_2}{\Norm{f(\tilde{x},u,d)}_2}$. Condition~2 of Lem.~\ref{lemma:rk2} therefore follows with $\sigma(s) =  L_{df} \cdot c_f \cdot s $.
\end{proof}